%% file: NoisyLinearInverse.tex
\documentclass[10pt]{article}

\usepackage{amsmath,amssymb,algorithm,amsthm,,cite,cases,bm,float,subfigure,graphicx,url}
\usepackage{pbox}
\numberwithin{equation}{section}
\newtheorem{defn}{Definition}

\newtheorem{propo}{Proposition}

\newtheorem{thm}{Theorem}
\newtheorem{lem}{Lemma}

\newtheorem{cor}{Corollary}

\newcommand{\ratio}{\frac{\sqrt{m}}{\sqrt{m-1}}}

\usepackage{amsmath,amssymb,algorithm,amsthm,,cite,subfigure,caption,cases,bm,float,subfigure,graphicx,url,color}
\usepackage[hidelinks]{hyperref}
\usepackage{comment}
\usepackage{thm-restate, enumerate }

\definecolor{darkred}{RGB}{150,0,0}
\definecolor{darkgreen}{RGB}{0,150,0}
\definecolor{darkblue}{RGB}{0,0,200}
\hypersetup{colorlinks=true, linkcolor=darkred, citecolor=darkgreen, urlcolor=darkblue}

\newcommand{\Keywords}[1]{\par\noindent
{\small{\em \noindent{\bf{Keywords}}\/}: #1}}

\usepackage[sc]{mathpazo} 
\usepackage[T1]{fontenc} 
\usepackage{microtype} 

\usepackage[hmarginratio=1:1,top=32mm,columnsep=20pt]{geometry} 
\usepackage{multicol} 
\usepackage{hyperref} 

\usepackage{lettrine} 
\usepackage{paralist} 

\usepackage{abstract} 

\usepackage{titlesec} 
\titleformat{\section}[block]{\large\scshape\centering}{\thesection.}{1em}{} 
\titleformat{\subsection}[block]{\large}{\thesubsection.}{1em}{} 

\title{\vspace{-5mm}\selectfont\textbf{Simple Bounds for Noisy Linear Inverse Problems\\ with Exact Side Information}\vspace{2pt}} 

\author{Samet Oymak \hspace{35pt}Christos Thrampoulidis\hspace{35pt} Babak Hassibi\vspace{6pt}\\Department of Electrical Engineering \\ Caltech, Pasadena -- 91125
\thanks{Email: {\{soymak,cthrampo,hassibi\}@caltech.edu}. This work was supported in part by the National Science Foundation under grants CCF-0729203, CNS-0932428 and CIF-1018927, by the Office of Naval Research under the MURI grant N00014-08-1-0747, and by a grant from Qualcomm Inc.}
}
\date{}

\begin{document}

\maketitle 

\input{commands}
\vspace{-10pt}
\begin{abstract} \vspace{-3pt}This paper considers the linear inverse problem where we wish to estimate a structured signal $\x_0$ from its corrupted observations. When the problem is ill-posed, it is natural to associate a convex function $f(\cdot)$ with the structure of the signal. For example, $\ell_1$ norm can be used for sparse signals. To carry out the estimation, we consider two well-known convex programs: 1) Second order cone program (SOCP), and, 2) Lasso. Assuming Gaussian measurements, we show that, 
if precise information about the value $f(\x_0)$ or the $\ell_2$-norm of the noise is available, one can do a particularly good job at estimation. In particular, the reconstruction error 
 becomes proportional to the ``sparsity'' of the signal rather than to the ambient dimension of the noise vector. We connect our results to the existing literature and provide a discussion on their relation to the standard least-squares problem. Our error bounds are non-asymptotic and sharp, they apply to arbitrary convex functions and do not assume any distribution on the noise.
 \vspace{4pt}
\end{abstract}

\input{introAndResults}

\input{simplified_lasso}

\input{proof}

\input{Biblio}
\end{document}

%% file: commands.tex
\newcommand{\beq}{\begin{equation}}
\newcommand{\eeq}{\end{equation}}
\newcommand{\bea}{\begin{align}}
\newcommand{\eea}{{\end{align}}}

\newcommand{\vp}{\vspace{4pt}}
\newcommand{\M}{\mathbf{M}}
\newcommand{\tC}{\tilde{\mathbf{C}}}
\newcommand{\tD}{\tilde{\mathbf{D}}}
\newcommand{\Fh}{\mathbf{\hat{F}}}
\newcommand{\ft}{\mathbf{\tilde{f}}}
\newcommand{\Z}{\mathbf{Z}}
\newcommand{\hz}{\hat{\z}}
\newcommand{\W}{\mathbf{W}}
\newcommand{\U}{\mathbf{U}}
\newcommand{\V}{\mathbf{V}}
\newcommand{\F}{\mathbf{F}}
\newcommand{\Delxf}{\boldsymbol{\omega}(\hat{T}_f(\x_0))}
\newcommand{\Pelxf}{{\mathbf{P}}_f(\x_0,{\mathbb{R}}^+)}
\newcommand{\Celxf}{{\mathbf{C}}_f(\x_0,{\mathbb{R}}^+)}
\newcommand{\Dlf}{{\mathbf{D}}_f(\x_0,\la )}
\newcommand{\order}[1]{\mathcal{O}\left(#1\right)}

\newcommand{\Dmin}{{\mathbf{D}}^*_f}
\newcommand{\Plf}{{\mathbf{P}}_f(\x_0,\la )}
\newcommand{\Df}{{\mathbf{D}}_f(\x_0,}
\newcommand{\Pf}{{\mathbf{P}}_f(\x_0,}
\newcommand{\Clf}{{\mathbf{C}}_f(\x_0,\la )}
\newcommand{\Cf}{{\mathbf{C}}_f(\x_0,}
\newcommand{\halp}{\hat{\alpha}}
\newcommand{\hkap}{\hat{\kappa}}

\newcommand{\siglong}{\sigma_{min}(\A,T_f(\x_0)\cap \Sc^{n-1})}
\newcommand{\DC}{\boldsymbol{\omega}(\hat{T}_f(\x_0))}
\newcommand{\DCC}{\boldsymbol{\omega}(\mathcal{\Cc}\cap\Bc^{n-1})}
\newcommand{\TAC}{\hat{T}_f(\x_0)}
\newcommand{\Tco}{\A{T}_f(\x_0)}
\newcommand{\CC}{{\mathbf{C}}(\Cc)}
\newcommand{\PC}{{\mathbf{P}}(\Cc)}

\newcommand{\TCB}{\hat{\Cc}}

\newcommand{\lac}{\lambda_{\text{crit}}}
\newcommand{\labb}{\lambda_{\text{best}}}
\newcommand{\taub}{\tau_{\text{best}}}
\newcommand{\lam}{\lambda_{{\max}}}

\newcommand{\DelxC}{{\mathbf{\Delta}}_\Cc(\x_0)}
\newcommand{\Gb}{\mathbf{G}}
\newcommand{\Sb}{\mathbf{S}}
\newcommand{\X}{\mathbf{X}}
\newcommand{\A}{\mathbf{A}}
\newcommand{\Lb}{\mathbf{L}}
\newcommand{\Y}{\mathbf{Y}}

\newcommand{\bu}{\text{Proj}}
\newcommand{\corr}{\text{corr}}
\newcommand{\bi}{{\Pi}}
\newcommand{\w}{\mathbf{w}}
\newcommand{\heta}{\hat{E}}
\newcommand{\geta}{\hat{\zeta}}
\newcommand{\feta}{{E^*}}
\newcommand{\x}{\mathbf{x}}
\newcommand{\ub}{\mathbf{u}}
\newcommand{\g}{\mathbf{g}}
\newcommand{\vb}{\mathbf{v}}
\newcommand{\q}{\mathbf{q}}
\newcommand{\p}{\mathbf{p}}
\newcommand{\bb}{\mathbf{b}}
\newcommand{\e}{\mathbf{e}}
\newcommand{\fh}{{\hat{f}}}
\newcommand{\tb}{\mathbf{t}}
\newcommand{\y}{\mathbf{y}}
\newcommand{\s}{\mathbf{s}}
\newcommand{\z}{\mathbf{z}}
\newcommand{\bseta}{{\boldsymbol{\eta}}}
\newcommand{\Iden}{\mathbf{I}}
\newcommand{\cb}{\mathbf{c}}  
\newcommand{\ab}{\mathbf{a}}
\newcommand{\oneb}{\mathbf{1}}
\newcommand{\h}{\mathbf{h}}
\newcommand{\hx}{{\mathbf{x}}^{*}(\la,\z)}
\newcommand{\fx}{\hat{\mathbf{x}}(\la,\z)}
\newcommand{\hw}{{\mathbf{w}}^{*}(\la,\z)}
\newcommand{\wh}{\hat{\mathbf{w}}}
\newcommand{\fw}{\hat{\mathbf{w}}(\la,\z)}
\newcommand{\ffw}{\hat{\mathbf{w}_r}(\la,\z)}
\newcommand{\xn}{{\mathbf{x}}_0}

\newcommand{\tih}{\tilde{\mathbf{h}}}
\newcommand{\tz}{\tilde{\z}}
\newcommand{\tw}{\tilde{\w}}
\newcommand{\bT}{{\bar{T}}}
\newcommand{\bP}{{\bar{P}_\w}}
\newcommand{\wP}{{{P}_\w}}
\newcommand{\bS}{{\bar{S}}}

\newcommand{\Sc}{{\mathcal{S}}}
\newcommand{\map}{{\text{map}}}
\newcommand{\call}{{\text{calib}}}
\newcommand{\Rca}{{\mathcal{R}}}
\newcommand{\Rco}{{\mathcal{R}}_{\text{ON}}}
\newcommand{\Rcf}{{\mathcal{R}}_{\text{OFF}}}
\newcommand{\Rci}{{\mathcal{R}}_{\infty}}
\newcommand{\Rr}{\mathbf{\rho}}
\newcommand{\Bc}{{\mathcal{B}}}
\newcommand{\Yc}{\mathcal{Y}}
\newcommand{\Dc}{\mathcal{D}}
\newcommand{\Lc}{\hat{\mathcal{L}}}
\newcommand{\Lco}{\mathcal{{L}}}

\newcommand{\wo}{\mathbf{\w}^*}

\newcommand{\Fco}{\mathcal{F}}
\newcommand{\Uco}{\mathcal{{U}}}
\newcommand{\fone}{L}
\newcommand{\Ldev}{\mathcal{L}_{dev}}
\newcommand{\Xc}{\mathcal{X}}
\newcommand{\Zc}{\mathcal{Z}}
\newcommand{\Uc}{\hat{\mathcal{U}}}
\newcommand{\pol}{^\circ}
\newcommand{\Kc}{\mathcal{K}}
\newcommand{\Nc}{\text{Null}}
\newcommand{\Rc}{\text{Range}}
\newcommand{\Nn}{\mathcal{N}}
\newcommand{\Cc}{\mathcal{C}}
\newcommand{\Gc}{\mathcal{G}}
\newcommand{\Ac}{\mathcal{A}}
\newcommand{\tac}{{\tilde{\mathcal{A}}}}
\newcommand{\Pc}{\mathcal{P}}
\newcommand{\fac}{\tilde{f}}
\newcommand{\R}{\mathbb{R}}
\newcommand{\Pro}{{\mathbb{P}}}
\newcommand{\E}{{\mathbb{E}}}

\newcommand{\xu}{\hat{x}_i}
\newcommand{\xz}{x_{0,i}}
\newcommand{\wi}{\hat{w}_i}

\newcommand{\ti}{\tilde}
\newcommand{\el}{\eta_{DN}(\la,\z)}
\newcommand{\paf}{\pa f(\x_0)}
\newcommand{\pef}{\pa f(\x)}
\newcommand{\dir}{D_{f,\x}(\w)}
\newcommand{\din}{D_{f,\x_0}(\w)}
\newcommand{\paw}{\pa f'_{\x_0}(\w)}
\newcommand{\fp}{f'_{\x_0}}
\newcommand{\ff}{\hat{f}_{\x_0}}
\newcommand{\ffx}{\hat{f}_{\x}}
\newcommand{\ffr}{f_{r,\x_0}}
\newcommand{\ffrx}{f_{r,\x}}
\newcommand{\la}{{\lambda}}
\newcommand{\lab}{{\bm{\lambda}}}
\newcommand{\eps}{\epsilon}
\newcommand{\om}{\omega}
\newcommand{\si}{\sigma}
\newcommand{\st}{\star}
\newcommand{\Fc}{\hat{{\mathcal{F}}}}
\newcommand{\Fch}{\hat{{\mathcal{F}}}}
\newcommand{\Tc}{\mathcal{T}}
\newcommand{\pa}{\partial}

\newcommand{\mub}{\boldsymbol{\mu}}
\newcommand{\sg}{\text{sgn}}
\newcommand{\hg}{\hat{g}}
\newcommand{\Du}{\text{dual}}
\newcommand{\cl}{\text{Cl}}
\newcommand{\tr}{\text{trace}}
\newcommand{\rk}{\text{rank}}
\newcommand{\cn}{\text{cone}}
\newcommand{\dbt}{\text{\bf{d}}}
\newcommand{\dt}{\text{{dist}}}
\newcommand{\dtR}{\dt_{\R^+}(\h)}
\newcommand{\vs}{\vspace}
\newcommand{\hs}{\hspace}
\newcommand{\has}{\hat{\s}}
\newcommand{\nn}{\nonumber}
\newcommand{\li}{\left<}
\newcommand{\ri}{\right>}
\newcommand{\nor}{\|\cdot\|}

%% file: introAndResults.tex
\section{Introduction}
\vspace{-3pt}
Second order cone programming (SOCP) and the Lasso are two common approaches to perform noise robust model fitting. They are often used for sparse approximation when the signal that underlies the observations is known to have few nonzero entries \cite{CanTao,WainLasso, OracleLAS,Tikh, BicRit,ZhaoLasso,DonLasso,Mon,BayMon,HighLAS}. This work considers the abstract model fitting problem where the signal has some sort of structure and we wish to estimate it from corrupted observations. To accomplish this, 
we use an abstract structure inducing convex function $f(\cdot)$. Let $\x_0\in\R^n$ be the true signal to be estimated. We observe $\y=\A\x_0+\z$ where $\A\in\R^{m\times n}$ is the measurement matrix and $\z$ is the noise vector. Let us now introduce the two problems mentioned above, the SOCP and the Lasso.
\vspace{-5pt}
\subsection{Lasso with exact side information}
Lasso is introduced by Tibshirani in \cite{Tikh}. The standard Lasso problem solves,
\beq\label{eq:LASSOgen}
\x_L^*=\arg\min_\x \la f(\x)+\frac{1}{2}\|\y-\A\x\|^2.
\eeq
In the program above and in the sequel, $\|\cdot\|$ is the $\ell_2$-norm. For the sake of this work, we assume that we know \emph{a priori} the value of the structure inducing function $f(\cdot)$ at $\x_0$. Under this information, we can simplify the problem to the following constrained setup,
\beq
\x_L^*=\arg\min_\x \frac{1}{2}\|\y-\A\x\|^2~~~\text{subject to}~~~f(\x)\leq f(\x_0)\label{LASSO}.
\eeq

\subsection{SOCP with exact side information}
SOCP is the name given to a class of algorithms. For linear inverse problems, a commonly used instance is the following \cite{CanTao},
\beq
\x_S^*=\arg\min_\x f(\x)~~~\text{subject to}~~~\|\y-\A\x\|\leq \delta.\nn
\eeq
Here $\delta$ is a known upper bound on the noise level $\|\z\|$. This ensures that the unknown  signal $\x_0$ is feasible for the SOCP. In this work, we will assume the exact information of $\|\z\|$ and solve,
\beq
\x_S^*=\arg\min_\x f(\x)~~~\text{subject to}~~~\|\y-\A\x\|\leq \|\z\|\label{SOCP}.
\eeq
\noindent In summary,
\begin{itemize}
\item Lasso will assume the knowledge about the signal, $f(\x_0)$. 
\item SOCP will assume the knowledge about the noise, $\|\z\|$.
\end{itemize}
We additionally assume that the sensing matrix $\A$ has independent zero-mean, $\frac{1}{m}$ variance Gaussian entries. Our main result provides non-asymptotic and sharp upper bounds on the estimation error terms $\|\x^*_L-\x_0\|$ and $\|\x^*_S-\x_0\|$. When $\x_0$ is a sparse vector and if we pick $f(\cdot)$ to be the $\ell_1$ norm, it is now well-known that the estimation error can be as small as $\|\z\|$. In this paper, we  restrict our attention to problems \eqref{LASSO} and \eqref{SOCP} and we try to answer the following three questions:
\begin{itemize}
\item[-] Can we generalize the results on $\ell_1$ norm to arbitrary convex functions?
\item[-] Can we give very sharp bounds with small and accurate constants?
\item[-] Can we do these non-asymptotically, i.e., for possibly very small number of measurements and/or sparsity levels?
\end{itemize}

\section{Result}
We will first state the general result and will consider specific examples later on. Let us  introduce the ``Gaussian width'' of a set. This concept is crucial for the statement of our results.
\begin{defn}[Gaussian width] Let $\Cc\in\R^n$ be a nonempty set. The Gaussian width of $\Cc$ is denoted by $\boldsymbol{\omega}(\Cc)$ and is defined as,
\beq
\boldsymbol{\omega}(\Cc)=\E\left[\sup_{\vb\in\Cc}\li\vb,\g\ri\right]\nn,
\eeq
where $\g\in\R^n$ has independent standard normal entries.
\end{defn}

Next, we require the definition of the tangent cone of a function $f(\cdot)$ at some $\x\in\mathbf{R}^n$. For this definition, let $\text{cone}(\cdot)$ and $\text{Cl}(\cdot)$ return the conic hull and the closure of a set, respectively.

\begin{defn}[Tangent cone] Assume $f(\cdot):\R^n\rightarrow\R$ and $\x\in\R^n$. Denote the set of descend directions $\{\vb\in\R^n\big| f(\x+\vb)\leq f(\x)\}$ by $D_f(\x)$. The tangent cone of $f(\cdot)$ at $\x$ is denoted by $T_f(\x)$ and defined as,
\beq
T_f(\x):=\text{Cl}(\text{cone}(D_f(\x))\nn.
\eeq
\end{defn}
Let $\Bc^{n-1}$ denote the unit $\ell_2$-ball in $\R^n$. For convenience, denote 
\begin{align}
\hat{T}_f(\x):=T_f(\x)\cap \Bc^{n-1}.\nn
\end{align}
  Finally, given a vector $\g\in\R^d$ with independent standard normal entries, we define $\gamma_d:=\E[\|\g\|]$. It is well known that $\gamma_d=\sqrt{2}\frac{\Gamma(\frac{d+1}{2})}{\Gamma(\frac{d}{2})}$ and $\sqrt{d}\geq \gamma_d\geq \frac{d}{\sqrt{d+1}}$ (see \cite{Cha}). This definition will simplify our notation in what follows. We are now ready to state our main result.

\begin{thm}\label{thm1}Consider the Lasso and SOCP problems in \eqref{LASSO} and \eqref{SOCP}, respectively. Assume $\z\in\R^m,\x_0\in\R^n$ are arbitrary and $\A\in\R^{m\times n}$ has independent $\Nn(0,\frac{1}{m})$ distributed entries. Assume $m\geq 2$ and $0\leq t\leq \gamma_m-\DC$. Then, with probability, $1-6\exp(-\frac{t^2}{26})$, we have,
\begin{align}
&\bullet~\|\x^*_L-\x\|\leq \eta(\x_0,t)\|\z\|\label{lasso1},\\
&\bullet~\|\x^*_S-\x\|\leq2 \eta(\x_0,t)\|\z\|\label{socp1},
\end{align}
where $\eta(\x_0,t)=\frac{\sqrt{m}}{\gamma_{m-1}}\frac{\DC+t}{\gamma_m-\DC-t}$.
\end{thm}
\noindent{\bf{Remark 1:}} Observing that $\gamma_{m-1}\gamma_m=m-1$ and $\gamma_{m-1}\leq\sqrt{m-1}$ leads to the bound, $\eta(\x_0,t)\leq \ratio\frac{\DC+t}{\sqrt{m-1}-\DC-t}$.

\noindent {\bf{Remark 2:}} In Theorem \ref{thm1}, we require $\gamma_m\geq \DC$. It has been shown that, this is indeed necessary, \cite{NoiseSense,McCoy}. When $\gamma_m<\DC$, it is futile to expect noise robustness, as one cannot perfectly recover $\x_0$ from \emph{noiseless} observations $\y=\A\x$ (cf. Theorem 3.4 of \cite{Oym}).

Our bound is \emph{only} in terms of the Gaussian width; which has been the subject of several works \cite{Mackey,Oym,Cha,McCoy,TotVar,SquareDeal}. This makes it possible to apply Theorem \ref{thm1} for specific choices of $f(\cdot)$ and $\x_0$ previously studied in the literature.


%% file: simplified_lasso.tex
\section{State-of-the-art applications}

We will now state our results for specific signal choices by making use of the existing results in the literature that compute upper bounds on the Gaussian width term $\DC$.

\noindent $\bullet$ {\bf{Sparse signals:}} When $\x_0$ is a $k$-sparse signal and $f(\cdot)$ is the $\ell_1$ norm, we have $\Delxf\leq \sqrt{2k\log\frac{2n}{k}}$, \cite{Cha}. Hence, we have the following.

\begin{cor} \label{for sparse}Suppose $\x_0$ is a $k$-sparse signal and $0\leq t\leq \sqrt{m-1}-\sqrt{2k\log\frac{2n}{k}}$. Pick $f(\cdot)$ to be the $\ell_1$ norm. Then, with probability $1-6\exp(-\frac{t^2}{26})$,
\beq
\|\x^*_L-\x_0\|\leq \|\z\|\ratio\frac{\sqrt{2k\log\frac{2n}{k}}+t}{\sqrt{m-1}-{\sqrt{2k\log\frac{2n}{k}}}-t}.\nn
\eeq
\end{cor}

\noindent $\bullet$ {\bf{Low-rank matrices:}} Nuclear norm (sum of the singular values) is the standard choice to encourage a low-rank solution. Suppose $\x_0$ is a rank-$r$ matrix of size $d\times d$. For this choice, it is known that $\Delxf\leq \sqrt{3r(2d-r)}$, \cite{Cha}.
\begin{cor} Suppose $\x_0\in\R^{d\times d}$ is a rank-$r$ matrix and $0\leq t\leq \sqrt{m-1}-\sqrt{3r(2d-r)}$. Pick $f(\cdot)$ to be the nuclear norm. Then, with probability $1-6\exp(-\frac{t^2}{26})$,
\beq
\|\x^*_L-\x_0\|\leq \|\z\|\ratio\frac{\sqrt{3r(2d-r)}+t}{\sqrt{m-1}-\sqrt{3r(2d-r)}-t}.\nn
\eeq
\end{cor}

\begin{figure}
  \begin{center}
{\includegraphics[scale=0.23]{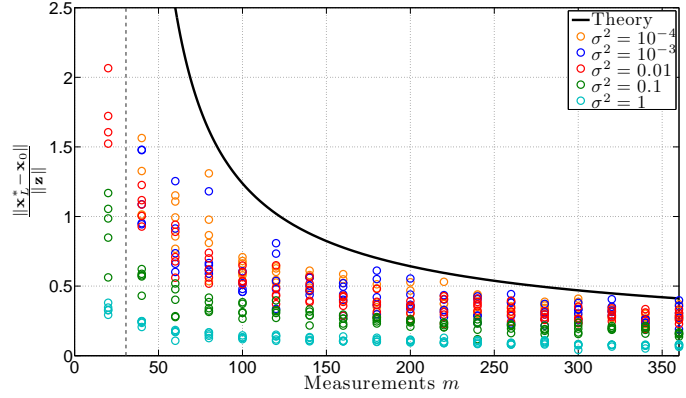}}
  \end{center}
  \caption{\footnotesize{We considered the sparse signal recovery setup of Corollary \ref{for sparse}. We set $n=500$, $k=5$ and varied $m$ from $0$ to $360$. Nonzero entries of $\x_0$ is generated with $\Nn(0,1)$ and then normalized to ensure unit norm. $\z$ and $\A$ has $\Nn(0,\sigma^2)$ and $\Nn(0,\frac{1}{m})$ entries respectively. Dashed line corresponds to the phase transition line $m=\Delxf^2$.}}
\label{figCont1}
\end{figure}

\noindent $\bullet$ {\bf{Block-sparse signals:}} Suppose the entries of $\x_0$ can be partitioned into $q$ known blocks of size $b$ and only $k$ of these $q$ blocks are nonzero. The standard function to encourage block-sparsity is the $\ell_{1,2}$ norm, which sums the $\ell_2$ norms of the individual blocks. For this choice, it is known that $\Delxf\leq \sqrt{4k(b+\log\frac{q}{k})}$, \cite{Mackey}.
\begin{cor} Suppose $\x_0\in\R^{qb}$ is a $k$-block-sparse signal and $0\leq t\leq \sqrt{m-1}-\sqrt{4k(b+\log\frac{q}{k})}$. Pick $f(\cdot)$ to be the $\ell_{1,2}$ norm. Then, with probability $1-6\exp(-\frac{t^2}{26})$,
\beq
\|\x^*_L-\x_0\|\leq \|\z\|\ratio\frac{\sqrt{4k(b+\log\frac{q}{k})}+t}{\sqrt{m-1}-\sqrt{4k(b+\log\frac{q}{k})}-t}.\nn
\eeq
\end{cor}
\noindent $\bullet$ {\bf{Other low-dimensional models}}: There are increasingly more signal classes that exhibit low-dimensionality and to which our results would apply. Some of these are as follows.
\begin{itemize}
\item Non-negativity constraint: $\x_0$ has non-negative entries, \cite{NonNeg}.
\item Low-rank plus sparse matrices: $\x_0$ can be represented as sum of a low-rank and a sparse matrix, \cite{LPS}.
\item Signals with sparse gradient: Rather than $\x_0$, its gradient ${\bf{d}}_{\x_0}(i)=\x_0(i)-\x_0(i-1)$ is sparse, \cite{TotVar}.
\item Low-rank tensors: $\x_0$ is a tensor and its unfoldings are low-rank matrices (see \cite{SquareDeal,tensorBen}).
\item Simultaneously sparse and low-rank matrices: For instance, $\x_0=\s\s^T$ for a sparse vector $\s$, \cite{OymSimult,Emile}.
\end{itemize} 
For more examples, the reader is referred to \cite{Cha,Negah,Oym,McCoy}.

\section{Interpretation of the results}
We will now argue that, one can easily interpret our results when the system $\y=\A\x_0+\z$ is seen as an $m\times \Delxf^2$ system rather than $m\times n$.
\subsection{Comparison to least squares} \label{least squares}
Consider the least-squares problem where one simply solves, 
\beq
\min_\x\|\y-\A\x\|\label{least squares}.
\eeq
It is clear that when $m<n$, \eqref{least squares} is hopeless and when $m>n$ and $\A$ has i.i.d. entries, $\A$ becomes full rank and the solution is $\x^*=(\A^T\A)^{-1}\A^T\y$. Hence, denoting the projection of $\z$ onto the range space of $\A$ by $\bu(\z,\text{Range}(\A))$ and the minimum singular value of $\A$ by $\sigma_{min}(\A)$,
\beq
\|\x^*-\x_0\|^2=\z^T\A(\A^T\A)^{-2}\A^T\z\leq \left(\frac{\|\bu(\z,\text{Range}(\A))\|}{\sigma_{min}(\A)}\right)^2,\label{inversion}
\eeq
It is well known that, when $\A$ has $\Nn(0,\frac{1}{m})$ entries, $\sigma_{min}(\A)\approx 1-\sqrt{\frac{n}{m}}$, \cite{Vers}. Also, since the range space is generated uniformly at random, $\|\bu(\z,\text{Range}(\A))\|\approx \sqrt{\frac{n}{m}}\|\z\|$. Consequently,
\beq
\|\x^*-\x_0\|\lesssim \|\z\|\frac{\sqrt{n}}{\sqrt{m}-\sqrt{n}}\label{ls1}.
\eeq
So, what is the relation between \eqref{ls1} and \eqref{lasso1}? Ignoring the $t$'s and using $\gamma_m\approx\sqrt{m}$ in \eqref{lasso1} , we find,
\beq
\|\x^*_L-\x\|\lesssim \|\z\|\frac{\DC}{\sqrt{m}-\DC}\label{lasso2}.
\eeq
One can move from \eqref{lasso2} to \eqref{ls1} by simply replacing the $\DC$ terms with $\sqrt{n}$. This indeed indicates that the Lasso and SOCP problems behave as $m\times \Delxf^2$ systems rather than $m\times n$ ones.
\subsection{Comparison to related works}
\noindent {\bf{Sparse recovery:}} 
A classical result states 
 that, when $\x_0$ is a sparse signal and when $\A$ has independent $\Nn(0,\frac{1}{m})$ entries the Lasso estimation error obeys $\order{\|\z\|\sqrt{\frac{{k\log n}}{m}}}$ when $m=\Omega({k\log\frac{n}{k}})$, \cite{BicRit,CandesDantzig,ellp,Negah,squareRoot}. Our bound given in Corollary \ref{for sparse} is fully consistent with this, however, we provide very small and accurate constants. 
In particular,  the phase transition occurring  around 
$2k\log\frac{2n}{k}$
  number of measurements shows up explicitly in our bound in Corollary \ref{for sparse} (see the term $\sqrt{m-1}-\sqrt{2k\log\frac{2n}{k}}$ in  the  denominator).\\
 

 
%
%
%



\noindent {\bf{Generalized linear inverse problems:}} 
Close to the present paper is the work due to \cite{Cha}. In \cite{Cha}, Chandrasekaran et al. perform error analysis of  the SOCP problem. Their result (cf. Corollary 3.3 in \cite{Cha}) shows that with probability $1-\exp(-\frac{1}{2}t^2)$,
\begin{align}\label{eq:ven}
\|\x_S^*-\x\| \leq 2\sqrt{m}\frac{\|\z\|}{\gamma_m-\DC-t}.
\end{align}
Our approach is related; however, we provide a more careful analysis. As a result of this, and in contrast to the error bound in \eqref{eq:ven} which grows linearly with the noise level $\|\z\|$, our bound \eqref{lasso1} is scaled by a constant factor of $\frac{\Delxf}{\sqrt{m}}$. This is due to the fact that we are able to carefully remove a significant component of the noise which cannot contribute to the error term.
\\

\noindent {\bf{Sharp error bounds for the Lasso estimator:}} There has been significant research interest in characterizing the error performance of the Lasso estimators.
\cite{Negah} provides a unified analysis of the error performance of the Lasso estimator \eqref{eq:LASSOgen}, which can be specialized to many regularizer functions.  More recent works establish \emph{sharper} bounds for the Lasso estimation error. In \cite{Mon,BayMon,NoiseSense}; Bayati, Montanari and Donoho provide explicit characterizations in an \emph{asymptotic} setting for $f(\cdot)=\|\cdot\|_{1}$. Closer in nature to the present paper, are the works \cite{StoLASSO} and \cite{Oym}.
The author in \cite{StoLASSO} analyzes the Lasso problem \eqref{LASSO} with prior information on $f(\x_0)$ when $f(\cdot)=\|\cdot\|_1$. \cite{Oym} generalizes the \emph{precise} analysis to arbitrary convex functions and, most importantly, extends it to penalized Lasso problems of the form \eqref{eq:LASSOgen}. Although tighter, the bounds in \cite{Oym} require stronger assumptions than ours, namely, an i.i.d. Gaussian noise vector $\z$ and an asymptotic setting where $m$ and $\Delxf$ is large enough. Their results translates to our framework as,
\beq
\|\x^*_L-\x\|\lesssim \|\z\|\frac{\DC}{\sqrt{m-\DC^2}}\label{lasso3}.
\eeq
The difference between \eqref{lasso2} and \eqref{lasso3} is in the denominator. $\sqrt{m-\DC^2}\geq \sqrt{m}-\DC$ for all regimes of $0\leq \DC^2<m$. The contrast becomes significant when $m\approx \DC^2$. In particular, setting $m=(1+\eps)^2\DC^2$, we have,
\beq
\frac{\sqrt{m-\DC^2}}{\sqrt{m}-\DC}=\frac{\sqrt{2\eps+\eps^2}}{\eps}=\sqrt{\frac{2}{\eps}+1}\nn.
\eeq
In summary, when $\eps$ is large, the bounds of this paper are as good as those of \cite{Oym,StoLASSO,BayMon,Mon}. When $\eps$ is small, they can be arbitrarily worse. Simulation results (see Figure \ref{figCont1}) verify that the error bounds of Theorem \ref{thm1} become sharp for large number of measurements $m$.  This difference can be intuitively explained by considering the least-squares error in \eqref{inversion}. There, using $\sigma_{min}(\A)$ as an upper bound results in a looser bound. For a vector $\z$ independent of $\A$, we actually have, 
\beq
\frac{n}{m-n}\|\z\|^2\approx \z^T\A(\A^T\A)^{-2}\A^T\z< \left(\frac{\|\bu(\z,\text{Range}(\A))\|}{\sigma_{min}(\A)}\right)^2\approx \left(\frac{\sqrt{n}}{\sqrt{m}-\sqrt{n}}\|\z\|\right)^2\nn.
\eeq
In this sense, \cite{Oym} considers the precise behavior of the left-hand side in \eqref{inversion} and we consider the looser bound given in the right-hand side; which makes use of the minimum singular value $\sigma_{min}(\A)$.

\section{Further remarks}
\subsection{Do we need the exact side information?}
For our results, we either assumed knowledge about the signal $f(\x_0)$, or knowledge about the noise $\|\z\|$. It is desirable to not be dependent on such quantities. A natural way to break this dependence is by using the following program,
\beq
\min_{\x} \la f(\x)+\frac{1}{2}\|\y-\A\x\|^2\label{real lasso}.
\eeq
When $f(\cdot)$ is the $\ell_1$ norm and $\x_0$ is a sparse signal, the problem becomes the original Lasso program introduced by \cite{Tikh} and it has been analyzed in great depth \cite{WainLasso, OracleLAS,Tikh, BicRit,ZhaoLasso}.
 Closer to us, Bayati and Montanari analyzes the precise noise characteristics of \eqref{real lasso} in \cite{BayMon,Mon}. Analysis of \eqref{real lasso} for the block-sparse signals and low rank matrices can be found in \cite{blocklasso,Koltc}. However, to the best of our knowledge, the existing guarantees are optimal up to a constant; while our bounds are almost exact.

While we leave the analysis of \eqref{real lasso} to a future work, we should emphasize that, \cite{Oym} proposed using,
\beq
\la=\frac{\|\z\|}{\sqrt{m}}\tau^*\sqrt{1-\frac{\Delxf^2}{m}},\label{best choice}\nn
\eeq
as the penalty parameter in \eqref{real lasso} and argued (non rigorously) that \eqref{real lasso} performs as good as \eqref{LASSO} with this choice. Here $\tau^*=\arg\min_{\tau\geq 0}\E[\dt(\g,\tau\paf)^2]$ where $\g\sim\Nn(0,\Iden_n)$ and $\dt(\g,\tau\paf)$ is the $\ell_2$-distance of the vector $\g$ to the $\tau$-scaled subdifferential $\tau\paf$. Similar choices has been proposed by various works for sparse recovery, \cite{BicRit,BayMon,Mon,WainLasso}. For sparse signals or low-rank matrices, $\tau^*$ only depends on sparsity (or rank) of the signal and has been the topic of several works \cite{Cha,Oym,McCoy,Mackey,BayMon,Mon}.
\subsection{Adversarial noise}

We will now consider the scenario where one has adversarial noise, i.e., noise has the information of the sensing matrix $\A$ and can adapt itself accordingly. In this case, the reconstruction error can become significantly worse. The following proposition illustrates this for the Lasso problem \eqref{LASSO}.
\begin{propo}\label{propo1} Assume $\x_0$ is not a minimizer of $f(\cdot)$. Then, given $\A\in\R^{m\times n}$ with independent $\Nn(0,\frac{1}{m})$ entries, with probability $1-\exp(-\frac{t^2}{2})$, there exists a noise vector $\z\in\R^m$ and Lasso optimum $\x^*_L$ such that, 
\beq
\|\x^*_L-\x_0\|\geq  \frac{\sqrt{m}}{\gamma_m+t}\|\z\|\nn
\eeq

\end{propo}
\begin{proof} Let $\x^*=\arg\min f(\x)$. Then, choose $\z=\A(\x^*-\x_0)$; which yields $\y=\A\x_0+\z=\A\x^*$. By construction, $\z\sim\Nn(0,\frac{\|\x^*-\x_0\|^2}{m}\Iden_m)$, hence with probability $1-\exp(-\frac{t^2}{2})$, $\|\z\|\leq (\gamma_m+t)\frac{\|\x^*-\x_0\|}{\sqrt{m}}$. Since $f(\x^*)\leq f(\x_0)$ and $\A\x^*-\y=0$, $\x^*$ is a (feasible) minimizer of \eqref{LASSO} and $\|\x^*-\x_0\|\geq \frac{\sqrt{m}}{\gamma_m+t}\|\z\|$.
\end{proof}
Proposition \ref{propo1} suggests that we can make error as big as the noise term $\|\z\|$. This contrasts with Theorem \ref{thm1} where the error is approximately $\frac{\Delxf}{\sqrt{m}}\|\z\|$ for sufficiently large $m$. The adversarial noise scenario can again be connected to least-squares in Section \ref{least squares}. In \eqref{inversion}, if the noise $\z$ already lies on $\text{Range}(\A)$, we will not have the reduction of $\sqrt{\frac{n}{m}}$ in the error. Similarly, Proposition \ref{propo1} constructs a noise vector that that lies in $\text{Range}(\A)$ and originates from the tangent cone element $\x^*-\x_0$. Hence, the resulting error norm is amplified by approximately $\frac{\sqrt{m}}{\Delxf}$.

Our next result gives an upper bound on the worst case error, which is close to the lower bound when $\Delxf\ll \gamma_m$. This uses a very similar argument to Corollary 3.3 of \cite{Cha}.
\begin{propo} Assume $\A\in\R^{m\times n}$ has independent $\Nn(0,\frac{1}{m})$ entries and assume $t<\gamma_m-\Delxf$. Then, with probability $1-\exp(-\frac{t^2}{2})$, the following bound uniformly hold for all noise vectors $\z\in\R^m$.
\beq
\max\{\|\x^*_L(\z)-\x_0\|,\|\x^*_S(\z)-\x_0\|\}\leq \frac{2\sqrt{m}\|\z\|}{\gamma_m-\Delxf-t}\nn
\eeq
\end{propo}
\begin{proof} From Lemma \ref{lemma:sing}, with probability $1-\exp(-\frac{t^2}{2})$, we have,
\beq
\min_{\vb\in T_f(\x_0)\cap\Sc^{n-1}}\|\sqrt{m}\A\vb\|\geq \gamma_m-\Delxf-t\nn
\eeq
Assuming this happens, we will show the result.

\noindent{{Proof for Lasso:}} $\x_0$ is feasible for \eqref{LASSO} hence $\|\y-\A\x^*_L\|\leq \|\z\|$. Also $\x^*_L-\x_0\in T_f(\x_0)$. Consequently,
\beq
\|\x^*_L-\x_0\|\frac{\gamma_m-\Delxf-t}{\sqrt{m}}-\|\z\|\leq \|\A(\x^*_L-\x_0)\|-\|\z\|\leq \|\A(\x^*_L-\x_0)-\z\|=\|\A\x^*_L-\y\|\leq \|\A\x_0-\y\|=\|\z\|.\label{lasso applies}
\eeq
{{Proof for SOCP:}} $\x_0$ is feasible for \eqref{SOCP} hence $f(\x^*_S)\leq f(\x_0)$ and $\|\y-\A\x^*_S\|\leq\|\z\|$ holds. Hence \eqref{lasso applies} will apply for $\x^*_S$ as well.

\end{proof}

%% file: proof.tex
\section{Proof of the Main Result}
We begin with introducing some necessary notation in Section \ref{sec:not}. In Section \ref{sec:pre}, we enlist two critical results for our analysis. Finally, Section \ref{main proof} provides the proof.

\subsection{Notation}\label{sec:not}
 Throughout the proofs, $\Tco$ will denote the cone obtained by multiplying elements of $T_f(\x_0)$ by $\A$., i.e.,
 $$
 \Tco = \{\A\vb\in\mathbb{R}^m ~|~ \vb\in T_f(\x_0)\}.
 $$
Let $\Cc\in\R^n$ be a convex subset of the unit $\ell_2$-sphere $\Sc^{n-1}$. Then, the minimum singular value of $\A\in\R^{m\times n}$ restricted to $\Cc$ is defined as,
\beq\nn
\sigma_{min}(\A,\Cc)=\min_{\vb\in\Cc}\|\A\vb\|.
\eeq
Observe that, $\sigma_{min}(\A,\Sc^{n-1})$  reduces to the standard definition of the minimum singular value of the matrix $\A$.
The projection of a vector $\vb\in\mathbb{R}^{n}$ onto a closed and convex set $\Cc$ is the unique vector $\bu(\vb,\Cc) = \arg\min_{\s\in\Cc}{\|\vb-\s\|}$. 

When $\Cc$ is a closed and convex cone, its polar is defined as $\Cc\pol=\{\ub\big|\ub^T\vb\leq 0,~\text{for all}~\vb\in \Cc\}$. Moreau's Decomposition Theorem \cite{More}, says that, any vector $\vb$ can be decomposed as,
\beq
\vb=\bu(\vb,\Cc)+\bu(\vb,\Cc\pol),~\text{where}~\li\bu(\vb,\Cc),\bu(\vb,\Cc\pol)\ri=0.\label{Moreaus}
\eeq

\subsection{Preliminary Results}\label{sec:pre}
The next lemma is due to Gordon \cite{Gor} and relates the Gaussian width to the restricted eigenvalue. This concept is similar to restricted isometry property and has been topic of several related papers, \cite{Cha,Negah,BicRit,Raskutti}.

\begin{lem}[Restricted eigenvalue] \label{lemma:sing} Let $\Gb\in\R^{m\times n}$ have independent standard normal entries and $\Cc\in\Sc^{n-1}$. Assume $0\leq t\leq\gamma_m-\boldsymbol{\omega}(\Cc)$. Then,
\beq\nn
\Pro\left(\min_{\vb\in\Cc}\|\Gb\vb\|\geq \gamma_m-\boldsymbol{\omega}(\Cc)-t\right)\geq 1-\exp(-\frac{t^2}{2}).
\eeq
\end{lem}
The next theorem is the main technical contribution of this work. It provides an upper bound on the correlation between a vector and elements of a cone multiplied by a Gaussian matrix.
\begin{thm} [Restricted correlation]\label{thm:main}
Let $\Cc\in\R^n$ be a convex and closed cone, $m\geq 2$ and $\z\in\R^m$ be arbitrary. Let $\Gb\in\R^{m\times n}$ have independent standard normal entries. For any $t\geq0$, pick $\alpha\geq \frac{\Delxf+t}{\gamma_{m-1}}\|\z\|$. Then,
\beq
\sup_{\vb\in\Cc\cap \Sc^{n-1}}\{\z^T\Gb\vb-\alpha\|\Gb\vb\|\}\leq 0,\label{max over v}
\eeq
with probability $1-5\exp(-\frac{t^2}{26})$.
\end{thm}

\subsection{Proof of Theorem \ref{thm1}}\label{main proof}
We will start by providing deterministic bounds on the estimation error. Then, with the help of Lemma \ref{lemma:sing} and Theorem \ref{thm:main}, we will finalize the proof.

\begin{lem}[Deterministic error bounds] \label{determine}Consider the problems \eqref{LASSO} and \eqref{SOCP}. We have,
\begin{align}
\max\{\|\x^*_L-\x_0\|,\frac{1}{2}\|\x^*_S-\x_0\|\}\leq \frac{\|\bu(\z,\Tco)\|}{\siglong}.\nn
\end{align}
\end{lem}
\begin{proof} Using \eqref{Moreaus}, let us write, $\z=\z_1+\z_2$ where $\z_1=\bu(\z,\Tco)$, $\z_2=\bu(\z,(\Tco)\pol)$, $\z_1^T\z_2=0$.

\noindent$\bullet$ \emph{Lasso:} Let $\w^*=\x^*_L-\x_0$. We will first show that $\|\A\w^*\|\leq\|\z_1\|$. Assume it is \emph{not} the case and let $\w'=\frac{\|\z_1\|}{\|\A\w^*\|}\w^*$. From convexity, $f(\x_0+\w')\leq f(\x_0)$, hence $\w'$ is feasible. We will show that $\|\z-\A\w'\|<\|\z-\A\w^*\|$, which will contradict with the optimality of $\w^*$.
\beq
\|\z-\A\w^*\|^2\geq \|\z-\A\w'+\A(\w'-\w^*)\|^2= \|\z-\A\w'\|^2+2\li\z-\A\w',\A(\w'-\w^*)\ri+\|\A(\w'-\w^*)\|^2\nn
\eeq
Now, observe that,
\begin{align}
\li\z-\A\w',\A(\w'-\w^*)\ri&\geq -\|\z_1\|\|\A(\w'-\w^*)\|-\li\A\w',\A(\w'-\w^*)\ri\nn\\
&\geq  -\|\z_1\|\|\A(\w'-\w^*)\|+\|\A\w'\|\|\A(\w'-\w^*)\|\nn\\
&\geq  (\|\A\w'\|-\|\z_1\|)\|\A(\w'-\w)\|>0\nn
\end{align}
Hence $\|\A\w^*\|\leq\|\z_1\|$. To conclude, we use the fact that $\|\w^*\|\leq \frac{\|\A\w^*\|}{\siglong}$.

\noindent$\bullet$ \emph{SOCP:} Let $\w^*=\x^*_S-\x_0$. Then, the problem becomes,
\beq
\w^*=\arg\min_\w f(\x_0+\w)~~~\text{subject to}~~~\|\z-\A\w\|\leq\delta\nn
\eeq
First observe that $0$ is feasible, hence $\w^*\in T_f(\x_0)$. Then, for any $\w\in T_f(\x_0)$,
\beq
\|\z\|^2=\|\z-\A\w\|^2=\|\z_2+\z_1-\A\w\|^2=\|\z_2\|^2+\|\z_1-\A\w\|^2+2\li\z_2,\z_1-\A\w\ri\geq \|\z_2\|^2+\|\z_1-\A\w\|^2,\nn
\eeq
where we used the fact that $\z_2^T\A\w\leq 0$ as $\A\w\in \Tco$. Now, using $\w^*\in T_f(\x_0)$, we find,
\beq
(\|\A\w^*\|-\|\z_1\|)^2\leq \|\z_1\|^2\implies \|\A\w^*\|\leq 2\|\z_1\|\implies \|\w^*\|\leq \frac{2\|\z_1\|}{\siglong}\nn
\eeq
\end{proof}

We are now ready to prove Theorem \ref{thm1}
\begin{proof}[Proof of Theorem \ref{thm1}]
Suppose $0\leq t<\gamma_m-\Delxf$. We will make use of the fact that, the following events hold with probability $1-\exp(-\frac{t^2}{2})-5\exp(-\frac{t^2}{26})$.
\begin{itemize}
\item Observe that $\boldsymbol{\omega}(T_f(\x_0)\cap\mathcal{S}^{n-1})\leq \Delxf$. Hence, applying Lemma \ref{lemma:sing} with $\Gb=\sqrt{m}\A$ and $\Cc=T_f(\x_0)\cap \Sc^{n-1}$, with probability $1-\exp(-\frac{t^2}{2})$, we have,
\beq
\siglong\geq \frac{\gamma_m-\Delxf-t}{\sqrt{m}}.\label{bound11}
\eeq
\item Applying Theorem \ref{thm:main} with $\A=\frac{\Gb}{\sqrt{m}}$ and $\Cc=T_f(\x_0)$, with probability $1-5\exp(-\frac{t^2}{26})$,
\beq
\|\bu(\z,\Tco)\| \leq \frac{\Delxf+t}{\gamma_{m-1}}\|\z\|.\label{bound12}
\eeq
To see this, pick $\vb$ in \eqref{max over v} such that $\A\vb=\frac{\bu(\z,\Tco)}{\|\bu(\z,\Tco)\|}$, which gives $\z^T\A\vb=\|\bu(\z,\Tco)\|$.
\end{itemize}
Now, the bounds in \eqref{lasso1} and \eqref{socp1} follow when we substitute \eqref{bound11} and \eqref{bound12} in Lemma \ref{determine}.
\end{proof}

\section{Proof of Theorem \ref{thm:main}}

\subsection{Auxiliary results}
There are a few ingredients of the proof. First, we require a result, which allows us to compare two Gaussian processes. This result is again due to Gordon (see Lemma 3.1 in \cite{Gor}). We make use of a slightly modified version of the original lemma, which can be found in \cite{Oym} (cf. Lemma 5.1).
\begin{lem} [Comparison Lemma, \cite{Gor}]\label{OTH} Let ${\bf{G}}\in\R^{m\times n}$ have independent standard normal entries. Let $\h\sim\Nn(0,\Iden_m)$ and $\g\sim\Nn(0,\Iden_n)$. Let $\Phi_1\subset\R^n$ be an arbitrary set and let $\Phi_2\subset\R^m$ be a compact set. Then,
\beq\nn
\Pro\left(\min_{\x\in\Phi_1}\max_{\ab\in\Phi_2} ~\x^T{\bf{G}}\ab~\geq c\right)\geq 2\Pro\left(\min_{\x\in\Phi_1}\max_{\ab\in\Phi_2}~\|\x\|\h^T\ab-\|\ab\|\g^T\x~\geq c\right)-1.
\eeq
\end{lem}
A function $f(\cdot):\R^n\rightarrow\R$ is called $L$-Lipschitz, if for all $\x,\y\in\R^n$,
\beq
|f(\x)-f(\y)|\leq L\|\x-\y\|\nn.
\eeq
The next lemma is a standard result on concentration properties of Lipschitz functions of Gaussian vectors, \cite{Tal}.
\begin{lem} \label{talagrand}Let $\g\sim\Nn(0,\Iden_n)$, $g\sim\Nn(0,1)$ and $f(\cdot):\R^n\rightarrow\R$ be an $L$-Lipschitz function.
Then, for $t\geq 0$,
\begin{align}
&\Pro(f(\g)-\E[f(\g)]\geq t)\leq \exp(-\frac{t^2}{2L^2})\nn,\\
&\Pro(f(\g)-\E[f(\g)]\leq -t)\leq \exp(-\frac{t^2}{2L^2})\nn.\\
&\Pro(g\geq t)\leq \frac{1}{2}\exp(-\frac{t^2}{2})\nn
\end{align}
\end{lem}
The following lemma provides a useful identity for the projection of a vector onto a cone.
\begin{lem} \label{super sym}Let $\Cc\subset\R^n$ be a closed and convex cone and $\vb\in\R^n$. Then,
\beq\nn
\max_{\ub\in\Cc\cap\Bc^{n-1}}\ub^T\vb=\|\bu(\vb,\Cc)\|.
\eeq
\end{lem}
\begin{proof}
From \eqref{Moreaus}, we have $\vb=\bu(\vb,\Cc)+\bu(\vb,\Cc\pol)$, where $\li\bu(\vb,\Cc),\bu(\vb,\Cc\pol)\ri=0$. For any $\ub\in\Cc$, $\ub^T\bu(\vb,\Cc\pol)\leq 0$, hence,
$\ub^T \vb \leq \ub^T \bu(\vb,\Cc)$. Since $\ub\in\Bc^{n-1}$, we further find from the Cauchy-Schwarz inequality that $\ub^T \vb\leq\|\bu(\vb,\Cc)\|$. On the other hand, picking $\ub=\frac{\bu(\vb,\Cc)}{\|\bu(\vb,\Cc)\|}\in\Cc\cap\Bc^{n-1}$, achieves $\ub^T \vb=\|\bu(\vb,\Cc)\|$.
%
\end{proof}


\subsection{Proof}
\begin{proof}[Proof of Theorem \ref{thm:main}] When $\z=0$, the problem is trivial, hence, assume $\z\neq 0$. If $\alpha\geq \|\z\|$, we clearly have,
\beq\nn
\sup_{\vb\in\Cc\cap \Sc^{n-1}}\{\z^T\Gb\vb-\alpha\|\Gb\vb\|\}\leq \sup_{\vb\in\Cc\cap \Sc^{n-1}}\{\|\z\|\|\Gb\vb\|-\alpha\|\Gb\vb\|\}\leq 0.
\eeq
Hence, without loss of generality, we may assume $\frac{\Delxf+t}{\gamma_{m-1}}\|\z\|\leq \alpha<\|\z\|$ and $t<\gamma_{m-1}-\Delxf$. Define the set $\Sc_\z=\alpha \Sc^{m-1}-\z$ and let $\TCB:=\Cc\cap \Bc^{n-1}$. Under this notation,
\beq\nn
\min_{\vb\in\Cc\cap \Sc^{n-1}} \alpha\|\Gb\vb\|-\z^T\Gb\vb=\min_{\vb\in\Cc\cap \Sc^{n-1}}\max_{\ub\in\Sc_\z}\ub^T\Gb\vb.
 \eeq
With this $\min\max$ formulation, we can apply Lemma \ref{OTH} and use the fact that $\|\vb\|=1$ to find,
\beq
\Pro\left(\min_{\vb\in \Cc\cap \Sc^{n-1}}\max_{\ub\in \Sc_\z}~\ub^T\Gb\vb\geq 0\right)\geq 2\Pro\left(\min_{\vb\in \Cc\cap \Sc^{n-1}}\max_{\ub\in \Sc_\z}\h^T\ub-\|\ub\|\g^T\vb\geq 0\right)-1\label{double prob},
\eeq
where $\h\sim\Nn(0,\Iden_m)$ and $\g\sim\Nn(0,\Iden_n)$. For the rest of the proof  we focus on the analysis of the \emph{simpler} optimization problem on the right hand side of \eqref{double prob}. Begin by noting that $\Cc\cap\Sc^{n-1}\subset\TCB$, hence,
\begin{align}\nn
\min_{\vb\in \Cc\cap \Sc^{n-1}}\max_{\ub\in \Sc_\z}~\h^T\ub-\|\ub\|\g^T\vb~\geq~\min_{\vb\in \TCB}\max_{\ub\in \Sc_\z}~\h^T\ub-\|\ub\|\g^T\vb.
\end{align}
The only term in which $\vb$ appears above is $\g^T\vb$. From Lemma \ref{super sym}, $\max_{\vb\in\hat{\Cc}}\g^T\vb=\|\bu(\g,\Cc)\|$. Hence, we find,
\beq
\min_{\vb\in \TCB}\max_{\ub\in \Sc_\z}~\h^T\ub-\|\ub\|\g^T\vb~=~\max_{\ub\in \Sc_\z}\left\{\h^T\ub-\|\ub\|\|\bu(\g,\Cc)\|\right\}.\nn
\eeq
Now, we make the change of variable $\ub=\alpha\ab-\z$ and write the right-hand side above as,
\beq
\max_{\ub\in \Sc_\z}\left\{\h^T\ub-\|\ub\|\|\bu(\g,\Cc)\|\right\}=\max_{\ab\in \Sc^{m-1}}\left\{\h^T(\alpha\ab-\z)-\|\alpha\ab-\z\|\|\bu(\g,\Cc)\|\right\}\label{524}
\eeq
Recall from \eqref{double prob}, that we want to lower bound the optimization problem above. The choice of $\ab$ is up to us and a good choice will guarantee a good lower bound on the right hand side of \eqref{524}. Let $\hz:=\frac{\z}{\|\z\|}$. Further, denote the projection of $\h$ onto $\z$ as $\h_2:=\hz\hz^T\h$. Also, $\h_1:=\h-\h_2$ and $\h_1$ is, by construction, orthogonal to $\z$ and is independent of $\h_2$. Let us choose 
$$\ab=\sqrt{1-\beta^2}\frac{\h_1}{\|\h_1\|}+\beta\hz,$$
 with
 $$\beta=\min\left\{\frac{\|\bu(\g,\Cc)\|}{\|\h_1\|},1\right\}.$$
 Then,
\begin{align}\nn
\|\alpha\ab-\z\|^2=\|\alpha\sqrt{1-\beta^2}\frac{\h_1}{\|\h_1\|}-\hz(\|\z\|-\alpha\beta)\|^2=\alpha^2(1-\beta^2)+(\|\z\|-\alpha\beta)^2,
\end{align}
and, denote,
\beq\nn
\kappa_1:=\sqrt{\alpha^2(1-\beta^2)+(\|\z\|-\alpha\beta)^2}\|\bu(\g,\Cc)\|
\eeq
Similarly,
\beq\nn
(\alpha\ab-\z)^T\h=\alpha\sqrt{1-\beta^2}\|\h_1\|-(\|\z\|-\alpha \beta){\h_2^T\hz},
\eeq
and, denote,
$$\kappa_2:=\alpha\sqrt{1-\beta^2}\|\h_1\| \quad\text{ and }\quad \kappa_3:=(\|\z\|-\alpha \beta){\h_2^T\hz}.$$ 
Note that $\g,\h_1,\h_2$ are all independent of each other and individually appears in $\kappa_1,\kappa_2,\kappa_3$ respectively.
From \eqref{524}, we are interested in $\Pro(\kappa_2-\kappa_1-\kappa_3\geq 0)$. Let $\halp=\frac{\alpha}{\|\z\|}$. To lower bound this, we will consider the events,
\beq
\|\h_1\|\geq \gamma_{m-1}-\tau;~~~\|\bu(\g,\Cc)\|\leq \DCC+\tau;~~~\h_2^T\hz\leq \tau\label{events}
\eeq
for some $\tau>0$ (to be determined) and the associated probabilities which are obtained as an application of Lemma \ref{talagrand}.
\begin{itemize}
\item $\Pro(\|\h_1\|\geq \gamma_{m-1}-\tau)\geq 1-\exp(-\frac{\tau^2}{2})$. 
\item $\Pro(\|\bu(\g,\Cc)\|\leq \DCC+\tau)\geq 1-\exp(-\frac{\tau^2}{2})$. 
\item $\Pro(\h_2^T\hz\leq \tau)\geq 1-\frac{1}{2}\exp(-\frac{\tau^2}{2})$
\end{itemize}

The first one holds from the fact that $\ell_2$-norm is $1$-Lipschitz. Second one follows from $1$-Lipschitzness of distance to a convex set \cite{Bertse}. Finally, the third bound follows from the fact that $\h_2^T\hz$ is statistically identical to $\Nn(0,1)$.

From the initial assumptions $\gamma_{m-1}-\DCC>t$. For the rest of the discussion, let $\tau=\frac{t}{3.6}$ and assume the three events in \eqref{events} hold, which happens with probability $1-\frac{5}{2}\exp(-\frac{\tau^2}{2})\geq 1-\frac{5}{2}\exp(-\frac{t^2}{26})$. We will now show that $\kappa_2-\kappa_1-\kappa_3\geq 0$. First, observe that, we have the following list of inequalities.
\begin{align}
\beta&=\frac{\|\bu(\g,\Cc)\|}{\|\h_1\|}\leq \frac{\DCC+\tau}{\gamma_{m-1}-\tau}\leq \frac{\DCC+2.6\tau}{\gamma_{m-1}-\tau}\nn\\&\leq \frac{\DCC+3.6\tau}{\gamma_{m-1}}\leq \frac{\alpha}{\|\z\|}<1.\label{my beta}
\end{align}
Also, since $\|\h_1\|\geq \gamma_{m-1}-\tau$,
\beq
(\frac{\alpha}{\|\z\|}-\beta)\|\h_1\|\geq \DCC+2.6\tau-\|\bu(\g,\Cc)\|\geq 1.6\tau\label{my beta 2}
\eeq
Let us focus on $\kappa_2-\kappa_1$ and let $\halp=\frac{\alpha}{\|\z\|}$. We may write,
\beq
\frac{\kappa_2-\kappa_1}{\|\z\|}=\halp\sqrt{1-\beta^2}\|\h_1\|-\sqrt{\halp^2(1-\beta^2)+(1-\halp\beta)^2}\|\bu(\z,\Cc)\|\nn
\eeq
Further normalizing by $\|\h_1\|$, we find,
\beq
\hkap(\halp):=\frac{\kappa_2-\kappa_1}{\|\z\|\|\h_1\|}=\halp\sqrt{1-\beta^2}-\beta\sqrt{\halp^2(1-\beta^2)+(1-\halp\beta)^2}\nn
\eeq
Expanding $\hkap(\halp)$,
\beq
\hkap(\halp)=\halp\sqrt{1-\beta^2}-\beta\sqrt{1+\halp^2-2\halp\beta}\nn
\eeq
For $\hkap(\halp)$, we have the following result.
\begin{lem} \label{super nice}Let $\beta$ be same as in \eqref{my beta}. Then, for $1\geq \halp\geq\beta$, we have that $\hkap(\halp)\geq\sqrt{ \frac{1-\beta}{2}}(\halp-\beta)$.
\end{lem}
\begin{proof}
Observe that $\hkap(\beta)=0$. Using $0\leq \beta<1$ and differentiating with respect to $\halp$, for $\halp\geq\beta$,
\begin{align}
\hkap'(\halp)&=\sqrt{1-\beta^2}-\frac{\beta(\halp-\beta)}{\sqrt{1+\halp^2-2\halp\beta}}\label{first der}.
\end{align}
Differentiating one more time, we find,
\beq
\hkap''(\halp)=-\frac{\beta([1+\halp^2-2\halp\beta]-\halp(\halp-\beta)+\beta(\halp-\beta))}{(1+\halp^2-2\halp\beta)^{3/2}}=\frac{-\beta(1-\beta^2)}{(1+\halp^2-2\halp\beta)^{3/2}}\leq 0.\nn
\eeq
Since the second derivative is nonpositive, this means $\hkap'(\halp)$ is minimized at $\halp=1$ over the region $\beta\leq \halp\leq 1$. Consequently, for $1\geq \halp\geq\beta$, we have,
\beq
\hkap(\halp)\geq \hkap(\halp)-\hkap(\beta)\geq  (\halp-\beta)\hkap'(1)\label{1isworst}
\eeq
To find $\hkap'(1)$, set $\halp=1$ in \eqref{first der},
\beq
\hkap'(1)=\sqrt{1-\beta^2}-\frac{\beta(1-\beta)}{\sqrt{2-2\beta}}=\sqrt{1-\beta}(\sqrt{1+\beta}-\frac{\beta}{\sqrt{2}})\geq\sqrt{\frac{1-\beta}{2}}.\nn
\eeq
Here we used the fact that $\sqrt{1+\beta}-\frac{\beta}{\sqrt{2}}$ is minimized at $\beta=1$ over $0\leq \beta\leq 1$, which can be verified by differentiating. Substituting this in \eqref{1isworst}, we find the desired result.
\end{proof}

Now, applying Lemma \ref{super nice} and using \eqref{my beta 2}, we have,
\beq
\frac{\kappa_2-\kappa_1}{\|\z\|}= \|\h_1\|(\hkap(\halp)-\hkap(\beta))\geq \|\h_1\|(\halp-\beta)\sqrt{\frac{1-\beta}{2}}\geq 1.6\tau\sqrt{\frac{1-\beta}{2}}\label{events5}
\eeq
Finally, to bound $\kappa_3$, for $1\geq \halp\geq \beta$, we use $0\leq \|\z\|-\alpha\beta\leq \|\z\|(1-\beta^2)$. This gives
\beq
0\leq \frac{\kappa_3}{\|\z\|}\leq (1-\beta^2)\tau\nn
\eeq
Combining with \eqref{events5}, we find,
\beq
\frac{\kappa_2-\kappa_1-\kappa_3}{\|\z\|}\geq 1.6\tau\sqrt{\frac{1-\beta}{2}}-(1-\beta^2)\tau\geq0\nn
\eeq
Here, the nonnegativity of the right-hand side is equivalent to,
\beq
2.56\frac{1-\beta}{2}\geq (1-\beta^2)^2\iff 1.28\geq (1+\beta)(1-\beta^2)\nn
\eeq
Differentiating the $(1+\beta)(1-\beta^2)$ term, we find that, it is maximized at $\beta=\frac{1}{3}$ and is upper bounded by $\frac{32}{27}\leq 1.28$. In summary, we have shown that, with probability $1-\frac{5}{2}\exp(-\frac{t^2}{26})$ \eqref{events} hold with $\tau=\frac{t}{3.6}$, and we have, $\kappa_2-\kappa_1-\kappa_3\geq 0$; which also implies nonnegativity of right-hand side of \eqref{524}. Now, using \eqref{double prob}, we find the desired result.
\end{proof}

